\newif\ifproofs\proofstrue
\newif\ifprivate\privatefalse
\newif\ifok\oktrue
\newif\ifworking\workingtrue
\newif\ifanswered\answeredfalse
\newif\iflong\longfalse
\documentclass[a4paper,english]{llncs}
\usepackage{xifthen}
\usepackage{amsmath, amssymb}
\usepackage{courier,xspace}
\usepackage{hyperref}
\usepackage{varioref}
\usepackage[]{units}
\usepackage{stmaryrd}
\usepackage{mathabx}
\usepackage{url}
\usepackage{pdflscape}
\usepackage[dvipsnames]{xcolor}
\usepackage{cite}
\usepackage{makeidx}
\usepackage{enumerate}  
\usepackage{color}
\usepackage{soul}

\usepackage[draft,silent,margin]{fixme}

\fxsetup{theme=color}
\definecolor{fxtarget}{rgb}{0.8000,0.0000,0.0000}
\FXRegisterAuthor{ebj}{EBJ}{Einar}

\fxusetheme{color}

\FXRegisterAuthor{ms}{ems}{ms}     

\fxsetup{mode=multiuser}
\fxsetup{targetlayout=color}
\fxsetface{margin}{\tiny}

\fxsetup{status=draft}

\fxuseenvlayout{color}
\fxsetup{envface=\scriptsize\tt}  

\fxsetup{inlineface=\bfseries}
\fxsetup{envlayout=color}

\newcommand{\casefont}            {\it}
\newcounter{casecounter}
\renewcommand{\thecasecounter}{\arabic{casecounter}}

\newcommand{\Casestring}          {Case}

\newcommand{\Subcasestring}       {Subcase}
\newcommand{\Subsubcasestring }   {Subsubcase}
\newcommand{\Basecasestring}      {Base step}
\newcommand{\Inductioncasestring} {Induction step}

\newenvironment{namedcase}[1]{
  \smallbreak \pagebreak[3]
  \noindent #1 \\}%

\newcommand{\intersect}        {\cap}
\newcommand{\angles}[1]        {\langle #1 \rangle}
\newcommand{\dom}              {\mathit{dom}}

\newcommand{\p}                {\vdash}                  

\newcommand{\of}               {{:}}                     
\newcommand{\OF}               {\mathrel{:}}             
\newcommand{\bnfdef}            {::=}
\newcommand{\bnfbar}            {\mid}

\newcommand{\scong}       {\equiv}         

\newcommand{\effect}[2]{\langle #1,#2\rangle} 
\newcommand{\ok}{\kw{ok}}

\newcommand{\tacins}           {\kw{in}}
\newcommand{\tacouts}          {\kw{out}}
\newcommand{\tacopens}         {\kw{open}}
\newcommand{\tacin}[1]         {\tacins\: #1}
\newcommand{\tacout}[1]        {\tacouts\: #1}
\newcommand{\tacopen}[1]       {\tacopens\: #1}

\newcommand{\tickins}          {\kw{tick}?}
\newcommand{\tickouts}         {\kw{tick}!}

\newcommand{\kw}[1]           {\mbox{\texttt{\bfseries #1}}}
\newcommand{\0}               {\kw{0}}
\newcommand{\repls}           {!}
\renewcommand{\parallel}      {\mathrel{\mid}}

\newcommand{\meta}[1]{\overline{#1}}

\newcommand{\names}       {\mathit{names}}

\newcommand{\C}           {\mathcal{C}}

\newcommand{\Chole}         {[\cdot]}
\newcommand{\Cambwith}[1]   {\mathcal{C}[#1]}
\newcommand{\Cwithhole}   {\C\Chole}
\newcommand{\Cwith}[1]    {\C[#1]}

\newcommand{\tamb}[2]          {#1[#2]}           
\newcommand{\nubind}[2]        {(\nu #1)\, #2}

\newcommand{\tcreseat}      {\kw{c}}
\newcommand{\tctickin}      {\tickins}
\newcommand{\tcreseatfon}   {\gamma}   
\newcommand{\tcin}{\kw{in}\:}
\newcommand{\tcopen}{\kw{open}\:}
\newcommand{\tcout}{\kw{out}\:}

\makeatletter
\newcommand{\xrightarrowtriangle}[2][]{\ext@arrow 0359\xrightarrowtrianglefill@{#1}{#2}}
\newcommand{\xrightarrowtrianglefill@}{%
  \arrowfill@\relbar\relbar{\mathrel{\smash{\rightarrowtriangle}\vphantom{\rightarrow}}}%
}
\makeatother

\newcommand{\red}             {\rightarrowtriangle}

\newcommand{\nmany}[1]                       {\widetilde{#1}}    

\newcommand{\alabtick}                       {\kw{tick}}

\newcommand{\trans}[1]        {\mathrel{\xrightarrow{#1}}}
\newcommand{\wtrans}[1]       {\ifthenelse{\isempty{#1}}{\xRightarrow{\;\;\;\;\;}}{\xRightarrow{#1}}}

\newcommand{\transtick}[1]    {\trans{\alabtick}}

\newcommand{\Tambs}                    {T}     

\newcommand{\Tts}                      {\textit{tkn}}
\newcommand{\Ttla}                     {\textit{subs}}
\newcommand{\Tcap}                     {\textit{bnd}}
\newcommand{\Tres}                     {\textit{req}}
\newcommand{\Tprov}                     {\textit{prov}}
\newcommand{\Tresj}                    {\textit{cap}}

\newcommand{\Tleq}                     {\leq}

\newcommand{\Tamb}[3]                  {\angles{#1,#2,#3}}

\newcommand{\Tcur}                     {\kw{this}}

\newcommand{\E}                        {\Gamma}

\newcommand{\Ea}                       {\Gamma}
\newcommand{\Eequiv}                   {\sim}
\newcommand{\Ec}                       {\Delta}

\newcommand{\Eplus}                    {\oplus}

\newcommand{\Eleq}                     {\sqsubseteq}

\newcommand{\Eempty}                   {\emptyset}
\newcommand{\ewithout}                 {\setminus}   

\newcommand{\judgeold}[5]                 {#1;\, #4 \p #3 \OF #5;\, #2}

\newcommand{\sbarbson}[1]        {{\downarrow_{#1}}}
\newcommand{\sbarbsonwith}[2]    {{\downarrow_{#1}^{#2}}}

  \newenvironment{ruleset}{

\noindent\mbox{}\noindent\hrulefill
\vspace{-2mm}
\begin{displaymath}
  \begin{array}[b]{c}}{

\end{array}\end{displaymath}
\hrulefill\mbox{}
}

\newcommand{\ruleskip}   {\\ \\[-0.3em]}

\newcommand{\rn}[1]{\mbox{\textsc{#1}}}        

\makeatletter
\newif\if@qeded\global\@qededfalse
\def\proofsketch{\global\@qededfalse\@ifnextchar[{\@xproofsketch}{\@proofsketch}}
\def\endproofsketch{\if@qeded\else\qed\fi\endtrivlist}
\def\endproofsketch{\if@qeded\else\fi\endtrivlist}  
\def\qed{\unskip \hfill{\unitlength1pt\linethickness{.4pt}\framebox(6,6){}}
\global\@qededtrue}
\def\@proofsketch{\noindent\trivlist \item[\hskip 
 \labelsep\emph{Proof sketch.}]\ignorespaces}
\def\@xproofsketch[#1]{\trivlist \item[\hskip 
 \labelsep\emph{Proof sketch #1.}]\ignorespaces}
\makeatother

\newcommand{\andalso}         {\quad\quad}

\newcommand{\freeze}[1]     {\widecheck{#1}}
\newcommand{\unfreeze}[1]   {\widehat{#1}}
\newcommand{\dunfrozen}[1]  {#1}

\newcommand{\frozen}[1]     {\check{#1}}
\newcommand{\frozenornot}[1]{\bar{#1}}

\definecolor{applegreen}{rgb}{0.55, 0.71, 0.0}
 \definecolor{airforceblue}{rgb}{0.36, 0.54, 0.66}
\definecolor{bittersweet}{rgb}{1.0, 0.44, 0.37}

\newcommand{\hlcc}[1]{%
  \colorbox{bittersweet!25}{$\displaystyle#1$}}

\usepackage{scalerel,stackengine}
\stackMath
\newcommand\reallywidehat[1]{%
\savestack{\tmpbox}{\stretchto{%
  \scaleto{%
    \scalerel*[\widthof{\ensuremath{#1}}]{\kern.1pt\mathchar"0362\kern.1pt}%
    {\rule{0ex}{\textheight}}
  }{\textheight}%
}{2.4ex}}%
\stackon[-6.9pt]{#1}{\tmpbox}%
}

\title{Assumption-Commitment Types for Resource Management
in Virtually Timed Ambients}
\author{Einar Broch Johnsen, Martin Steffen, Johanna Beate Stumpf }
\institute{University of Oslo, Oslo, Norway\\\textsf{\{einarj,msteffen,johanbst\}@ifi.uio.no}}

\begin{document}

\pagestyle{plain}

\maketitle
\ifok 
\begin{abstract}
  This paper introduces a type system for resource management in the
  context of nested virtualization. With nested virtualization, virtual
  machines compete with other processes for the resources of their host
  environment in order to provision their own processes, which could again
  be virtual machines. The calculus of virtually timed ambients formalizes
  such resource provisioning, extending the capabilities of mobile ambients
  to model the dynamic creation, migration, and destruction of virtual
  machines. The proposed type system uses assumptions about the outside of
  a virtually timed ambient to guarantee resource provisioning on the
  inside. We prove subject reduction and progress for well-typed virtually
  timed ambients, expressing that upper bounds on resource needs are
  preserved by reduction and that processes do not run out of resources.
\end{abstract}


\section{Introduction}
\label{sec:tamb.intro}

Virtualization enables the resources of an execution environment to be
represented as a software layer, a so-called \emph{virtual
  machine}. Software processes are agnostic to whether they run on a
virtual machine or directly on physical hardware.  A virtual machine
is itself such a process, which can be executed on another virtual
machine.  Technologies such as VirtualBox, VMWare ESXi, Ravello HVX,
Microsoft Hyper-V, and the open-source Xen hypervisor increasingly
support running virtual machines inside each other in this way. This
\emph{nested virtualization}, originally introduced by Goldberg
\cite{goldberg74computer}, is necessary to host virtual machines with
operating systems which themselves support
virtualization~\cite{ben-yehuda10turtles}, such as Micro\-soft Windows
7 and Linux KVM.  Use cases for nested virtualization include end-user
virtualization for guests, software development, and deployment
testing.  Nested virtualization is also a crucial technology to
support the hybrid cloud, as it enables virtual machines to migrate
between different cloud providers~\cite{williams12xenblanket}.

To study the logical behavior of virtual machines in the context of
nested virtualization, this paper introduces a type-based analysis for
a calculus of virtual machines. An essential feature of virtual
machines, captured by this calculus, is that a virtual
machine competes with other processes for the resources available in
their execution environment, in order to provision resources to the
processes \emph{inside} the virtual machine.  Another essential
feature of virtual machines is \emph{migration}. From an abstract
perspective, virtual machines can be seen as mobile processes which
can move between positions in a hierarchy of nested locations.

We develop our type system for \emph{virtually timed ambients}
\cite{johnsen.steffen.stumpf:calculusvta-wadt}, a calculus of mobile
virtual locations with explicit resource provisioning, based on mobile
ambients \cite{cardelli.gordon:ambients}.  Our goal is to statically
approximate an upper bound on resource consumption for systems of
virtual machines expressed in this calculus.  The calculus features a
resource called \emph{virtual time}, reflecting local execution
capacity, which is provisioned to an ambient by its parent ambient,
similar to time slices that an operating system provisions to its
processes.  
With several levels of nested virtualization,
virtual time becomes a \emph{local} notion 
which depends on an
ambient's position in the location hierarchy.  Virtually timed
ambients are mobile, reflecting that virtual machines may migrate
between host virtual machines.  Migration affects the execution speed
of processes inside the virtually timed ambient which is moving as
well as in its host before and after the move.  Consequently, the
resources required by a process change dynamically when the topology
changes.

The distinction between the inside and outside of a virtually timed
ambient (or a virtual machine) is a challenge for compositional
analysis; we have knowledge of the current contents of the virtual
machine, but not of what can happen outside its borders. This
challenge is addressed in our type system by distinguishing
\emph{assumptions} about ambients on the outside of the virtually
timed ambient from \emph{commitments} to ambients on the inside.  To
statically approximate the effects of migration, an ambient's type
imposes a bound on the ambients it can host. Type checking fails if
the ability to provision resources for an incoming ambient in a timely
way cannot be statically guaranteed.

The ambient calculus has previously been enriched with types 
(e.g., \cite{G03}).  
Exploiting the explicit notion of resource provisioning in virtually
timed ambients (including a fair scheduling strategy and competition
for resources between processes), our type system captures the
resource capacity of a virtually timed ambient and an upper bound on
the number of its subambients. The type system thereby provides
concrete results on resource consumption in an operational framework.
Resource dependency in the type system is expressed using
coeffects. The term \emph{coeffect} was coined by Petricek, Orchard,
and Mycroft \cite{POM14,POM13} to capture how a computation
depends on an environment rather than how it affects the
environment. In our setting, coeffects capture how a process depends
on its environment by an upper bound on the resources needed
by the process.

\smallskip
\noindent \emph{Contributions.}
The main technical contributions of this paper are
\begin{itemize}
\item an \emph{assumption commitment type system} with \emph{effects}
  and \emph{coeffects}, which provides a static approximation of
  constraints regarding the capacity of virtually timed ambients and
  an upper bound on their resource usage; and
\item a proof of the \emph{soundness of resource management} for
  well-typed virtually timed ambients in terms of a \emph{subject
    reduction} theorem which expresses that the upper bounds on resources
  and on the number of subambients are preserved under reduction, and a
  \emph{progress} theorem which expresses that  well-typed  virtually timed ambients will
  not run out of resources.
\end{itemize}
To the best of our knowledge, this is the first  assumption commitment
style type system for resource types and nested locations. 

\smallskip
\noindent \emph{Paper overview.}  Section~\ref{sec:coeff.vta} introduces virtually timed
ambients. Section~\ref{sec:coeff.typesystem.simple} presents the 
type system for resource management.
In Section~\ref{sec:coeff.subred}, we prove the soundness of the type
system in terms of subject reduction and progress. We discuss related
work and conclude in Sections~\ref{sec:coeff.related} and
\ref{sec:coeff.conclusion}.


\fi
\section{Virtually Timed Ambients}

\label{sec:coeff.vta}

\ifok
Mobile ambients \cite{cardelli.gordon:ambients} are processes with a
concept of location, arranged in a hierarchy which may change
dynamically. Interpreting these locations as places of deployment,
\emph{virtually timed ambients}
\cite{johnsen.steffen.stumpf:calculusvta-wadt,johnsen.steffen.stumpf:calculusvta-jlamp}
extend mobile ambients with notions of virtual time and resource
consumption.  The timed behavior of a process depends on the one hand on
the \emph{local} timed behavior, and on the other hand on the placement or
deployment of the process in the hierarchical ambient structure.  Virtually
timed ambients combine timed processes and timed capabilities with the
mobility and location properties of the mobile ambient calculus.

Compared to the previous work
\cite{johnsen.steffen.stumpf:calculusvta-wadt,johnsen.steffen.stumpf:calculusvta-jlamp},
we here present a slightly simplified version of virtually timed
ambients which assumes a uniform speed for all ambients in the
hierarchy.
This simplification does not mean the ambients proceed uniformly
with respect to time: the progress of an ambient still depends on
its position in the hierarchy and the number of sibling ambients that
compete for time slices at the given level.  Since an ambient system can
change its structure, i.e., its hierarchy, an ambient's local access to
time slices may also dynamically change.
Thus, the simplification by uniform speed is not conceptual, but it
allows a simpler formulation of the type system by removing fractional
representations of speed in scheduling and the resulting (easy but
cumbersome) calculations.


\begin{definition}[Virtually timed ambients]
  \label{definition:tamb.timedambients}%
  \index{virtually timed ambient}
  The syntax of virtually timed ambients is as follows:
 {\em    \begin{displaymath}
    \begin{array}[t]{rcll}
  P & \bnfdef & \0 \bnfbar \nubind{n}{P} \bnfbar  P\!\parallel\!P\  \bnfbar\  \repls C.P  \ \bnfbar  C.P
  \bnfbar
  \tamb{n}{P}
  \\
  C & \bnfdef &  \tacin{n} \bnfbar \tacout{n} \bnfbar \tacopen{n} \bnfbar \tcreseat
\end{array}

    \end{displaymath}
}
\end{definition}

The syntax is almost unchanged from that of standard mobile ambients
(e.g.,~\cite{cardelli.gordon:ambients}), the only syntactic addition is an
additional capability $\tcreseat$ explained below.  In the sequel, we
mostly omit the qualification ``timed'' or ``virtually timed'' when
speaking about processes, capabilities, etc. Processes include the inactive
process $\0$, parallel composition $P\!\parallel\!P $ and replication
$\repls C.P$, the latter conceptually represents an unbounded parallel
composition of a process, with capability $C$ as ``guard''.  The
$\nu$-binder or restriction operator, makes the name $n$ local, as in the
$\pi$-calculus, ambient calculus and related formalisms.  \emph{Ambients}
$\tamb{n}{P}$ are named processes. The standard mobile ambient
\emph{capabilities} $\tacins$, $\tacouts$, and $\tacopens$ allow a process
to change the nested ambient structure by moving an ambient into or out of
another ambient, or by dissolving an ambient altogether.

The additional capability $\tcreseat$ is specific for the virtually timed
extension and abstractly represents the need of the process for a
\emph{resource} in order to continue its execution (i.e., $\tcreseat$ can
be read as ``consume'').  Thus, the consume capability relates to
computation cost in frameworks for cost analysis (e.g.,
\cite{albert07esop,albert18tocl}).  In our setting, the
$\tcreseat$-capabilities consume resources which can be thought of time
slices and which are governed by a scheduler. A scheduler is \emph{local}
to an ambient and its responsibility is to fairly schedule the processes
that are directly contained in the ambient it is managing. Since ambients
are nested, the scheduler also has to allocate time slices or resources to
subambients, thereby delegating the fair allocation of time slices at the
level of the subambients to their respective schedulers. To achieve a fair
schedule, the semantics adopts a simple round-based strategy. In first
approximation: no process is served twice, unless all other processes at
that level have been served at least once.  This \emph{round-based scheme}
is slightly more refined in that the number of processes per ambient is not
fixed as ambients may move inside the hierarchy and even dissolve.

To capture the outlined scheduling strategy in operational rules
working on the syntax of ambients, we \emph{augment} the grammar of
Def.~\ref{definition:tamb.timedambients} with additional
\emph{run-time} syntax (highlighted below).  When needed, we refer to
the original syntax from Def.~\ref{definition:tamb.timedambients} as
\emph{static} syntax.  The run-time syntax uses the notation
$\frozen{\_}$ to indicate that processes, including ambients, are
\emph{frozen} and $\meta{n}$ to denote either $n$ or $\frozen{n}$.
$$
  \begin{array}[t]{rcll}
    P & \bnfdef & \0\bnfbar \nubind{n}{P}  \parallel \ P \parallel  P\ \bnfbar\   \repls C.P\  \bnfbar\ \hlcc{\tickins }
    \bnfbar \hlcc{\tickouts}
    \bnfbar 
    \tamb{\meta{n}}{P} 
    \bnfbar
    C.P
    \\
    \meta{n} &\bnfdef &   \dunfrozen{n} \bnfbar \hlcc{\frozen{n}} 
    \\
    \tcreseatfon &\bnfdef & \dunfrozen{\tcreseat} \bnfbar \hlcc{\frozen{\tcreseat}}
    \\
    C & \bnfdef & 
    \tacin{n} \bnfbar \tacout{n} \bnfbar  \tacopen{n} \bnfbar \hlcc{\tickins} \bnfbar \tcreseatfon
  \end{array}

$$
Frozen processes are not eligible for scheduling.  For regular
(non-ambient) processes, only processes prefixed by the consume capability
$\tcreseat$ will be controlled in this way; other processes are
unconditionally enabled. Consequently, we only need as additional
run-time syntax $\frozen{\tcreseat}$, capturing a deactivated resource
capability. Similarly $\tamb{\frozen{n}}{P}$ denotes  a timed ambient
which is not eligible for scheduling.
Apart from scheduling, a frozen ambient $\tamb{\frozen{n}}{P}$ is
treated as any other ambient $\tamb{n}{P}$: the ordinary, untimed
capabilities address ambients by their name without the additional
scheduling annotation.  Likewise, $\nu$-binders and corresponding
renaming and algebraic equivalences treat names $\frozen{n}$ as
identical to $n$. Unless explicitly mentioned, we assume in the
following run-time syntax, i.e., $P$ may contain occurrences of
$\frozen{n}$ and $\frozen{\tcreseat}$. Time slices are denoted by
ticks, and come in two forms $\tickins$ and $\tickouts$. We may think
of the first form $\tickins$ as representing incoming ticks into an
ambient, typically from the parent ambient, the second form
$\tickouts$ represents time slices handed out to the local processes
by the local scheduler. The $\tickins$-capability similarly accepts an
incoming tick.  Let $\names(P)$ denote the set of names for ambients
contained in $P$.

\subsection{Semantics}
\label{sec:coeff.simple.semantics}

The semantics of virtually timed ambients is given as a reduction
system.  The rules for structural congruence $P \scong Q$ are
equivalent to those for mobile ambients (and therefore omitted here).
Besides structural congruence, the \emph{reduction} relation
$P \red{} Q$ for virtually timed ambients builds upon
\emph{observables,} also known as \emph{barbs}. Barbs, originally
introduced for the $\pi$-calculus
\cite{milner.sangiorgi:barbed-bisimulation}, capture a notion of
immediate observability. In the ambient calculus, these observations
concern the presence of a top-level ambient whose name is not
restricted \cite{merro.nardelli:behavioral}. In our context, the barbs
are adapted to express top-level schedulability, i.e., an ambient's
ability to receive a tick.  In addition, we will need to capture that
a sub-process is able to receive a tick from it's local scheduler. To
specify that, we denote by $\Cwithhole$ (or simply by $\C$) a
\emph{context}, i.e., a process with a (unique) hole $\Chole$ in place
of a process, and write $\Cambwith{P}$ for the context with its hole
replaced by $P$. The observability predicates (or ``tick-barbs'')
$\sbarbson{\tctickin}$ resp.\ $\sbarbsonwith{\tctickin}{\C}$ are then
defined as follows, where $\nmany{m}$ is a tuple of names:

\begin{definition}[Barbs]
  \label{definition:coeff.barbs}
  A process $P$ \emph{strongly barbs on $\tctickin$}, written
  $P\sbarbson{\tctickin}$, if $P\scong (\nu \nmany{m})
  (\tamb{n}{P_1} \parallel P_2)$ or $P\scong (\nu \nmany{m})
  (\tctickin.P_1 \parallel P_2)$.   A process $P$
  \emph{strongly barbs on} $\tctickin$ \emph{in context} $\C$, written
  $P\sbarbsonwith{\tctickin}{\C}$, if $P = \Cwith{P'}$ for some process $P'$
  with $P' \sbarbson{\tctickin}$.
\end{definition}

Note that the ambient name $n$ may well be hidden, i.e., mentioned in
$\nmany{m}$. Barbing on the ambient name $n$, written $P
\sbarbson{n}$, would require that $P\scong (\nu \nmany{m})
(\tamb{n}{P_1} \parallel P_2)$ where $n\notin \nmany{m}$, in
contrast to the definition of $P\sbarbson{\tctickin}$. This more
conventional notion of strong barbing \cite{merro.nardelli:behavioral}
expresses that an ambient is available
for interaction with the standard ambient capabilities; ambients whose name
is unknown are not available to be contacted by other ambients and
therefore, their name is excluded in the observability predicate
$\sbarbson{n}$. In contrast, strong barbing as defined in Def.~\ref{definition:coeff.barbs}
captures an ambient's ability to receive ticks and thus, the definition
will allow hidden ambients to be served by the
local scheduler. 
However, the name of the ambient must \emph{not} be frozen
$\frozen{n}$: ambients that have been served a tick in the current
round are not eligible for another allocation before a new round has
started, in which case the ambient's name has ``changed'' to $n$.

The reduction rules for virtually timed ambients are given in
Tables~\ref{tab:redsteps.simple.1} and \ref{tab:redsteps.simple.2}.  The
rules in Table~\ref{tab:redsteps.simple.1} (with rule names to the left)
cover ambient reconfiguration. Apart from the annotations used for
scheduling, the rules are exactly the ones from the (untimed) mobile
ambients \cite{cardelli.gordon:ambients}.

\begin{table}[t]
  \centering
  \begin{ruleset}
\begin{array}{l@{\quad}l}
  \rn{(R-In)} &
  \tamb{\meta{n}}{\tacin{m}.P_1\parallel P_2} \parallel \tamb{\meta{m}}{Q}
  \red{}
  \tamb{\meta{m}}{Q\ \parallel \tamb{\frozen{n}}{P_1\parallel P_2}}
\\
\rn{(R-Out)}&
  \tamb{\meta{m}}{\tamb{\meta{n}}{\tacout{m}.P_1\parallel P_2} \parallel Q}
  \red{}
   \tamb{\frozen{n}}{P_1\parallel P_2} \parallel \tamb{\meta{m}}{Q}
\\
\rn{(R-Open)} & 
\tacopen{n}.P_1 \parallel \tamb{\meta{n}}{P_2} 
  \red{}
  P_1 \parallel \freeze{P}_2
\end{array}  
  \end{ruleset}
  \caption{Reduction rules (1). \label{tab:redsteps.simple.1}}
\end{table}

Ambients can undergo restructuring in three different ways.
First, an ambient can move horizontally or laterally by entering a
sibling ambient (rule \rn{R-In}).
Second, it can move vertically up the tree, leaving its parent ambient
(rule \rn{R-Out}).
%
Finally, a process can cause the dissolution of its
surrounding ambient (rule \rn{R-Open}).
These forms of restructuring are \emph{timeless}
in that they incur no computation costs.
If an ambient changes its place, the scheduler of the target ambient
will from that point on, become responsible for the new ambient, and
the treatment is simple: if frozen, the newcomer will not be served in the
current round of the scheduler,
but waits for the next round.
Considering the source ambient (i.e., the ambient which contained the
process executing the $\tacouts$ or $\tacins$ capability), no
process inside the source ambient looses or changes its status.  A
similar discipline is followed when opening an ambient in rule
\rn{R-Open}.
Note that a process in an ambient can execute a capability $\tacins$,
$\tacouts$, or $\tacopens$ independent of the status of the affected
ambient, which is indicated in the rules by $\meta{n}$ and $\meta{m}$.

\ifprivate
\begin{remark}[Newcomers have to wait]
  This principle that newcomers have to line up at the end of the queue in
  the target ambient is consistent with the old rules. The condition for
  the source agent, that not even the continuation looses its time slice in
  the moving source agent deviates from the older rules.
\end{remark}
\fi

To realize the round-based scheduling, processes conceptually switch
back and forth between waiting to be served in the current round, and
having been served and thus waiting for the next round to begin. The
following definition of $\freeze{P}$ is used to mark a process $P$ as
served:

\begin{definition}[Freezing and unfreezing]
  \index{$\freeze{P}$}%
  \label{definition:vta.simple.freeze}
  Let $\freeze{P}$ denote the process where all top-level occurrences of
  $\tamb{n}{Q}$ are replaced by $\tamb{\frozen{n}}{Q}$ and all top-level
  occurrences of $\tcreseat$ replaced by $\frozen{\tcreseat}$. Conversely,
  let $\unfreeze{P}$ denote the process where all top-level occurrences
  of $\frozen{\tcreseat}$ are replaced by $\tcreseat$ and all top-level
  occurrences of $\tamb{\frozen{n}}{Q}$ replaced by $\tamb{n}{Q}$.
Define $\freeze{P}$
  by induction on the syntactic structure as follows:
$$
    \begin{array}[t]{rcl@{\quad}l@{\quad\quad}rcl}
  \freeze{\nubind{n}{P}} & = & \nubind{n}{\freeze{P}}  & & \freeze{\frozen{\tcreseat}} &  = & \frozen{\tcreseat}
  \\
  \freeze{P_1\parallel P_2} & = &  \freeze{P_1}\ \parallel \freeze{P_2} & &  \freeze{\tcreseat}&  = & \frozen{\tcreseat}
  \\
  \freeze{\tamb{n}{P}} & = & \tamb{\frozen{n}}{P}  &  &  \freeze{n}   & =  & \frozen{n}
  \\
  \freeze{\tcreseatfon.P} & = &  \frozen{\tcreseatfon}.P & & \freeze{\frozen{n}}   & =  & \frozen{n}
  \\
  \freeze{C.P} & = &  C.\freeze{P} &  \text{$C \not= \tcreseatfon$}
  \\
  \freeze{P}  & = & P & \text{otherwise}\\[-12pt]
\end{array}%
$$
  The definition of $\unfreeze{P}$ is analogous (e.g.,
  $\unfreeze{\frozen{\tcreseat}} = \tcreseat$) and omitted here.
\end{definition}

Remark that the congruence relation, which is part of the
reduction semantics, works with scheduling in the sense that both
operations defined in Def.~\ref{definition:vta.simple.freeze} are
preserved under congruence: $P_1 \scong P_2$ implies
$\freeze{P_1} \scong \freeze{P_2}$ and
$\unfreeze{P_1} \scong \unfreeze{P_2}$ .

\begin{table}[t]
  \centering
  \begin{ruleset}
    \begin{array}[t]{l@{\qquad\qquad}l}
\begin{array}[t]{l@{\qquad\qquad}l}
  \tickins \red{} \tickouts
  &
  \tickouts \parallel \tickins.P \red{} \freeze{P}
  \\
  \dunfrozen{\tcreseat}.P \red{} \tickins.P
  &
  \tickouts \parallel \tamb{\dunfrozen{n}}{P}\red{}
  \tamb{\frozen{n}}{\tickins \parallel P}
\end{array}
&
\begin{array}[t]{l}
 \text{not}(P\sbarbson{\tctickin})
\\\hline\\[-12pt]  
\tamb{n}{P} \red{} \tamb{n}{\unfreeze{P}}
\end{array}
\end{array}

  
  \end{ruleset}
  \caption{Reduction rules (2)\label{tab:redsteps.simple.2}}
\end{table}
 
Scheduling is covered by the reduction rules in
Table~\ref{tab:redsteps.simple.2}, which details the handling of ticks
and the resource capabilities.  The first rule translates ``incoming''
ticks to ticks available for local processes.
The translation ratio is uniform; i.e., one incoming tick produces one
outgoing tick (this is the simplification compared to previous work
mentioned earlier, where the ratio between incoming and local ticks
could more generally be a rational number).
A $\tickouts$ process can be consumed in two ways. First by scheduling
a $\tcreseat$-prefixed process which undergoes the steps
$\dunfrozen{\tcreseat}.P \red{} \tctickin.P \red{} \freeze{P}$
(consuming $\tickouts$ in the second step). Second, by scheduling a
subambient, such that an incoming tick $\tickins$ occurs one level
down in the hierarchy.
To ensure the round-based scheduling, the scheduled entity must
not have been served yet in the current round. For this purpose, the
process before the transition must be of the form
$\dunfrozen{\tctickin}.P$ or $\tamb{\dunfrozen{n}}{P}$, and after
the transition the continuation of the process is frozen,
using Def.~\ref{definition:vta.simple.freeze}. 
The last rule completes one scheduling round and initiates the next
round by changing the ambient's processes $P$ to $\unfreeze{P}$.  This
unfreezing step can be done only if all the ambient's processes have
been served, which is captured be the negative premise stipulating
that no process at the level of $n$ can proceed: at the given level,
the processes are blocked, but that does not mean, that in
subambients, all processes must be blocked as well.


\begin{example}
  \label{ex1}
  Consider the process $\tickouts \parallel
  \tamb{\textit{cloud\,}}{\0} \parallel \tamb{vm}{\tcin
    \textit{cloud\,}.\tcreseat. \0}$.  Three reduction steps are possible,
  as $\tickouts$ can propagate to either ambients and ambient $vm$
  can move into $\textit{cloud}$.  One way this process can reduce, is as
  follows:
  
  $$\begin{array}{l}
    \tickouts \parallel \tamb{\textit{cloud\,}}{\0}   \parallel  \tamb{vm}{\tcin \textit{cloud\,}.\tcreseat. \0}\\
\quad\red{} \tickouts \parallel  \tamb{\textit{cloud\,}}{\0 \parallel
    \tamb{\frozen{vm}}{\tcreseat. \0}}
\red{}  \tamb{\frozen{\textit{cloud\,}}}{\tickins \parallel \0 \parallel
    \tamb{\frozen{vm}}{\tcreseat. \0}}\\
     \quad\red{}  \tamb{\frozen{\textit{cloud\,}}}{\tickouts \parallel \0 \parallel
    \tamb{vm}{\tcreseat. \0}}
\red{}  \tamb{\frozen{\textit{cloud\,}}}{\0 \parallel
    \tamb{\frozen{vm}}{\tickins \parallel \tcreseat. \0}}\\
    \quad\red{}  \tamb{\frozen{\textit{cloud\,}}}{\0 \parallel
    \tamb{\frozen{vm}}{\tickouts \parallel \tcreseat. \0}}
\red{}  \tamb{\frozen{\textit{cloud\,}}}{\0 \parallel
    \tamb{\frozen{vm}}{\0}}\\
    \end{array}
$$
However, the time slice could  also enter the ambient \textit{vm}, and
move with this ambient, resulting in a reduction sequence
starting as follows:
$$\begin{array}{l}
    \tickouts \parallel \tamb{\textit{cloud\,}}{\0}   \parallel
    \tamb{vm}{\tcin \textit{cloud\,}.\tcreseat. \0}\\
    \red{}  \tamb{\textit{cloud\,}}{\0}
    \parallel  \tamb{\frozen{vm}}{\tickins \parallel \tcin \textit{cloud\,}.\tcreseat. \0}
\red{}  \tamb{\textit{cloud\,}}{\0 \parallel
    \tamb{\frozen{vm}}{\tickins \parallel \tcreseat. \0}} \red \ldots
\end{array}
$$

\end{example}


\section{An Assumption-Commitment Type System}
\label{sec:coeff.typesystem.simple}
\index{type system}

We consider a type system which analyzes the timed behavior of
virtually timed ambients in terms of the movement and resource
consumption of a given process. Statically estimating the timed
behavior is complicated because the placement of an ambient in the
process hierarchy influences its resource consumption, and movements
inside the hierarchy changes the relative speed of the ambients.  The
proposed type system is loosely based on Cardelli, Ghelli, and
Gordon's movement control types for mobile ambients \cite{CGD01};
however, its purpose is quite different, and therefore the technical
formulation will be rather different as well.

\paragraph{Types, contexts, and judgments.}
The typing of processes happens with respect to nominal resource
contracts for virtually timed ambients.  Contracts $\Tambs$ for
ambients are tuples of the form
\begin{equation*}
  \label{eq:coeff.types.ambient}
  \Tambs = \Tamb{\Tresj}{\Tcap}{\Tts}.
\end{equation*}
Here, 
$\Tresj \in \mathbb{N}$ specifies the ambient's \emph{resource
  capacity}, i.e., the upper bound on the number of resources that the
subprocesses of the ambient are allowed to require;
$\Tcap \in \mathbb{N}$ specifies the ambient's \emph{hosting
  capacity}, i.e., the upper bound on the number of timed subambients
and timed processes allowed inside this ambient; and
$\Tts \in \mathbb{N}$ specifies the ambient's \emph{currently hosted
  processes}, i.e., the number of taken slots within the ambient's hosting
capacity.
The number of currently hosted processes inside an ambient
can change dynamically, due to the movements of ambients. These
changes must be captured in the type system. In this sense, a type for
ambient names $\Tambs$ contains an accumulated effect mapping.

Typing \emph{environments} or contexts associate ambient names with
resource contracts. They are finite lists of associations of the form
$n:\Tambs$. In the type system, when analyzing an ambient or process,
a typing environment will play a role as an \emph{assumption},
expressing requirements about the ambients \emph{outside} the current
process. Dually, facts about ambients which are part of the current
process are captured in another typing environment which plays the
role of a \emph{commitment}. Notationally, we use $\Ea$ for assumption
and $\Ec$ for commitment environments.  We write $\Eempty$ for the
empty environment, and $\Ea, n: \Tambs$ for the extension of $\Ea$ by
a new binding $n :\Tambs$. We assume that ambient names $n$ are unique
in environments, so $n$ is not already bound in $\Ea$. Conversely,
$\Ea \ewithout n: \Tambs$ represents an environment coinciding with
$\Ea$ except that the binding for $n$ is removed. If $n$ is not
declared in $\Ea$, the removal has no effect.  The typing judgement
for names is given as $\Ea \p n : T$.  Since each name occurs at most
once, an environment $\Ea$ can be seen as a finite mapping; we use
$\Ea(n)$ to denote the ambient type associated with $n$ in $\Ea$ and
write $\dom(\Ea)$ for all names bound in~$\Ea$.  In the typing rules,
the typing environment $\Ea$ may need to capture the ambient in which
the current process resides; this ambient will conventionally be
denoted by the \emph{reserved name} $\Tcur$.

Typing judgements for \emph{processes} $P$ are of the form
$$  \label{eq:coeff.types.judgment.process}
  \judgeold{\Ea}{\Ec}{P}{\Tres}{\ok \effect{\Tprov}{\Ttla}} 
$$
where $\Tres$ and $\Tprov$ are the required and provided resources of
a process $P$, $\Ttla$ is the number of subambients of $P$, and $\Ea$
and $\Ec$ are the assumptions and commitments of $P$, respectively.
We call $\Tres$ the \emph{coeffect} of the
process. Coeffects~\cite{POM13,POM14} capture how a computation
depends on an environment rather than how it affects this
environment. We use the perspective of coeffects since a computation
may require resources from its environment to terminate. Similarly,
$\Tprov$ is the number of provided resources in $P$; these resources
are available in $P$ independent of its environment, and $\Ttla$
approximates the number of subambients in $P$. We may think of
$\effect{\Tprov}{\Ttla}$ as the \emph{effect} of the type judgment, where
effects express what the process $P$ potentially provides to its
environment.

Since ambient names are assumed to be unique, it follows for type
judgments that $\dom(\Ec) \intersect \dom(\Ea) = \emptyset$, as an
ambient is either inside the process and has its contract in the commitments,
or outside and has its contract in the assumptions. Further,
$\dom(\Ec) \subseteq\names(P)$.

\begin{definition}[Domain equivalence]
  \index{domain equivalence} \index{$\Eequiv$} Two contexts
  $\Ea_1$ and $\Ea_2$ are \emph{domain equivalent}, denoted $\Ea_1 \Eequiv
  \Ea_2$, iff $\dom(\Ea_1)= \dom(\Ea_2)$.
\end{definition}

For each process, the domain of the assumptions is assumed to
contain all names which are not in the domain of the commitments;
i.e., for two parallel processes $P_1$ and $P_2$ such that
$ \judgeold{\Ea_1}{\Ec_1}{P_1}{\Tres_1}{\ok\effect{\Tprov_1}{\Ttla_1}}$ and
$ \judgeold{\Ea_2}{\Ec_2}{P_2}{\Tres_2}{\ok\effect{\Tprov_2}{\Ttla_2}}$, 
we will have that $\Ec_2\subseteq\Ea_1$, $\Ec_1\subseteq\Ea_2$
and $\dom(\Ec_1) \intersect \dom(\Ec_2) = \emptyset$.

\begin{definition}[Additivity of contexts]\label{def_add}
  Let $\Ea_1$ and $\Ea_2$ be contexts such that $\Ea_1 \Eequiv \Ea_2$, and
  $\Ea_i(n) = \Tamb{\Tresj}{\Tcap}{\Tts_i}$ for $n\in dom (\Ea_1)$ and $i =
  1, 2$.
  The context $\Ea_1 \Eplus \Ea_2$ with domain $n\in dom (\Ea_1)$ is
  defined as follows:
  \begin{displaymath}
    (\Ea_1 \Eplus \Ea_2) (n) =
    \Tamb{\Tresj}{\Tcap}{\Tts_1 +\Tts_2}.
\end{displaymath}
\end{definition}

If the number of currently hosted ambients is smaller than the hosting capacity
of all ambients in an environment,
we say that the environment is error-free:

\begin{definition}[Error-free environments]
  An environment $\E$ is \emph{error-free}, denoted \emph{$\p \Ea :
    \ok$} if $\Tts\leq \Tcap$ for all $n\in dom(\Ea)$ and
  $\Ea(n) = \Tamb{\Tresj}{\Tcap}{\Tts}$.
\end{definition}

Resource contracts can be ordered by their contents and environments
by their resource contracts. The bottom type $\bot$ is a subtype of
all resource contracts.

\begin{definition}[Ordering of resource contracts and environments]
  \label{definition:coeff.ordering.uniform}
  \index{$\Tleq$ (order on resource contracts)} Let
$\Tambs_1 = \Tamb{\Tresj_1}{\Tcap_1}{\Tts_1}$ and $\Tambs_2 =
\Tamb{\Tresj_2}{\Tcap_2}{\Tts_2}$ be resource contracts.
Then $T_1$ is a \emph{subtype} of $T_2$, written $T_1 \Tleq T_2$, if
and only if $\Tresj_1 \leq \Tresj_2$, $\Tcap_1 \leq \Tcap_2$ and
$\Tts_1 \geq \Tts_2$.  Typing environments are \emph{ordered} by the
subtype relation as follows: For type environments $\Ea_1$ and
$\Ea_2$, $\Ea_1 \Eleq \Ea_2$ if $\dom(\Ea_1) \subseteq \dom(\Ea_2)$
and $\Ea_1(n)\ \Tleq\ \Ea_2(n)$, for all $n \in \dom(\Ea_1)$.
\end{definition}

Scheduling is reflected in the type rules by the
calculation of the coeffect $\Tres$. The coeffect captures the number of
resources a process needs to terminate.



In Table~\ref{tab:type.rules}, Rule~\rn{T-Zero} captures the
inactive process, which does not require nor provide any time slices.
Rule~\rn{T-Tick1} expresses the availability of $\tickouts$ and
Rule~\rn{T-Tick2} that a time slice $\tickins$ is ready to be
consumed. Both judgements express that a time slice is provided
without requiring any time slice.
The assumption rule \rn{T-Ass} types an ambient with the resource contract
it has in the environment.  The restriction rule \rn{T-Res} removes the
resource contract assumption in the environment for the restricted name.
Subsumption relates different resource contracts,
; e.g.,
in subtypes (\rn{T-Tsub}), the subsumption rule \rn{T-Sub} allows a higher
number of required resources, a lower number of provided resources and a
higher number of subambients to be assumed in a process.

For the typing of ambients in Rule~\rn{T-Amb}, the number of resources
a process $P$ requires changes if it becomes enclosed in an ambient
$n$; i.e., we move to the resource contract $T$ of $n$, provided the
process $P$ satisfies its part of the contract.
The contract here becomes a
commitment whereas the required resources in the co-effect may be
smaller than the $\Tcap$ of the contract because $n$ may already have
received the time slices $\Tprov$.

%
The parallel composition rule \rn{T-Par} makes use of the fairness of
the scheduling of time slices in virtually timed ambients. While the
branches agree on the required resources $\Tres$, the provided
resources and subambients accumulate.
It follows from \rn{T-Par} that several ambients in parallel will at
most need as many resources $\Tres$ from the parent ambient as the
slowest of them. 
Furthermore, \rn{T-Par} changes assumptions and commitments depending on
the assumptions and the commitments of the composed processes, using the
context composition operator from Def.~\ref{def_add} to compose
environments.  We have $\dom(\Ec_P) \intersect \dom(\Ec_Q) = \emptyset$,
which is a consequence of the uniqueness of ambient names.  The assumptions
of the branches split the resource contracts of the environment $\E$
between the type judgements for $P_1$ and $P_2$ and the commitments split
such that $\Delta_1'$ is the assumption for $P_1$ and vice versa.
%
For replication, the corresponding rule \rn{T-Rep} imposes the restriction,
that the process being replicated does not incur any cost; allowing that
would amount to an unbounded resource need.

\begin{table}[!t]
\begin{ruleset}
\begin{array}{c}
\textsc{\footnotesize (T-Zero)}\\\hline
\judgeold{\Eempty}{\Eempty}{\0}{0}{\ok \effect{0}{0}}
\end{array}

\quad

\begin{array}{c}
\textsc{\footnotesize (T-Tick1)}\\\hline
\judgeold{\Eempty}{\Eempty}{\tickins}{0}{\ok
    \effect{1}{0}}
\end{array}

\quad

\begin{array}{c}
\textsc{\footnotesize (T-Tick2)}\\\hline
\judgeold{\Eempty}{\Eempty}{\tickouts}{0}{\ok
    \effect{1}{0}}
\end{array}

\ruleskip

\begin{array}{c}
\textsc{\footnotesize (T-Ass)}\\
  \E(n)=  \Tambs
\\\hline
  \Ea \p  n  : \Tambs
\end{array}

\qquad

\begin{array}{c}
\textsc{\footnotesize (T-Res)}\\
  \judgeold{\Ea,k: \Tambs}{\Ec}{P}{\Tres}{\ok\effect{\Tprov}{\Ttla}}
\\\hline
  \judgeold{\Ea}{\Ec}{(\nu  k :\Tambs) P}{\Tres}{\ok\effect{\Tprov}{\Ttla}}
\end{array}

\qquad

\begin{array}{c}
\textsc{\footnotesize (T-Tsub)}\\
  \Ea \p  n  : \Tambs_1 \andalso
 \Tambs_1 \leq \Tambs_2
\\\hline
  \Ea \p  n  : \Tambs_2
\end{array}

\qquad


\ruleskip

\begin{array}{c}
\textsc{\footnotesize{ (T-Amb)}}\\
  \Gamma\vdash n: T \andalso 
  T=\langle\Tresj,\Tcap,\Tts \rangle 
  \\
  \Ttla \leq \Tcap
  \andalso 
  \Tres \times \Tcap\leq \Tresj+\Tprov
  \\
  \judgeold{\Ea,\Tcur:T}{\Ec}{P}{\Tres}{\ok
  \effect{\Tprov}{\Ttla}}
  \\\hline
  \judgeold{\Ea}{n\of T,\Ec}{\tamb{\frozenornot{n}}{P}}{\Tresj}{\ok\effect{0}{\Tcap+ 1}}  
\end{array}

\qquad

\begin{array}{c}
\textsc{\footnotesize (T-Sub)}\\
  \Ttla \geq \Ttla'
\\
    \Tres \geq \Tres'
  \andalso
\Tprov'\geq  \Tprov
\\
  \judgeold{\Ea}{\Ec}{P}{\Tres'}{\ok\effect{\Tprov'}{\Ttla '}}
\\\hline
  \judgeold{\Ea}{\Ec}{P}{\Tres}{\ok\effect{\Tprov}{\Ttla}}
\end{array}

\ruleskip

\begin{array}{c}
\textsc{\footnotesize{ (T-Par)}}\\
  \Ec_1 \Eequiv \Ec_1' 
  \andalso 
  \Ec_2 \Eequiv \Ec_2' 
  \andalso
  \p \Ea : \ok
  \andalso 
  \p \Ec : \ok
  \\
  \Ea = \Ea_1 \Eplus \Ea_2
  \andalso\quad
  \Ea_1 \Eequiv \Ea_2
  \andalso
  \judgeold{\Ea_1,\Ec_2'}{\Ec_1}{P_1}{\Tres}{\ok\effect{\Tprov_1}{\Ttla_1}}
\\
  \Ec = (\Ec_1 \Eplus \Ec_1') ,(\Ec_2 \Eplus \Ec_2')
  \andalso 
  \judgeold{\Ea_2,\Ec_1'}{\Ec_2}{P_2}{\Tres}{\ok\effect{\Tprov_2}{\Ttla_2}}
\\\hline
 \judgeold{\Ea}{\Ec}{P_1\parallel P_2}{\Tres}{\ok\effect{\Tprov_1+\Tprov_2}{\Ttla_1 + \Ttla_2}}
\end{array}

 \ruleskip



\begin{array}{c}
  \textsc{\footnotesize (T-Consume1)}\\
\Ttla'=\max\{\Ttla,1\}\\
  \Ea; \Tres\vdash P : \ok \effect{\Tprov}{\Ttla}, \Ec
\\\hline
\Ea; \Tres+1 \vdash \tcreseat.P : \ok \effect{\Tprov}{\Ttla'}, \Ec
\end{array}

\qquad

\begin{array}{c}
  \textsc{\footnotesize (T-Consume2)}\\
\Ttla'=\max\{\Ttla,1\}\\
  \Ea; \Tres\vdash P : \ok \effect{\Tprov}{\Ttla}, \Ec
\\\hline
\Ea; \Tres+1 \vdash \tickins.P : \ok \effect{\Tprov}{\Ttla'}, \Ec
\end{array}

\ruleskip

\begin{array}{c}
\textsc{\footnotesize (T-In)} \\
   T=\langle \Tresj, \Tcap ,\Tts \rangle 
 \andalso
  T'=\langle \Tresj, \Tcap ,\Tts+\Tcap\,'+1 \rangle\\
  \E,m\of T; \Tres\vdash P: \ok \effect{\Tprov}{\Ttla}, \Delta
\andalso
 \Tcap\times \Tres\leq \Tresj
\\ 
  \Gamma\vdash \Tcur : \langle \Tresj',\Tcap\,' ,\Tts' \rangle
\andalso
  \Tts+\Tcap\,'+1 \leq  \Tcap
  \\\hline
\judgeold{\E,m\of T'}{\Ec}{\tcin m.P}{\Tres}{\ok \effect{\Tprov}{\Ttla}}
\end{array}

\quad

 \begin{array}{c}
\quad\\
\textsc{\footnotesize (T-Rep)}\\
   \Ea; 0 \p P \OF 0, \Ec_P
   \\ 
   C \in  \{\tcin n, \tcout n, \tcopen n\}
\\\hline
     \Ea; 0  \p  ! C.P: 0, \Ec_P 
\end{array}

\ruleskip

\begin{array}{c}
\textsc{\footnotesize (T-Out)}\\
  \E;
  \Tres\vdash P: \ok \effect{\Tprov}{\Ttla}, \Delta   
\\\hline  
  \E; \Tres \vdash \tcout m.P  :\ok \effect{\Tprov}{\Ttla}, \Delta
\end{array}

\andalso
\begin{array}{c}
  \textsc{\footnotesize (T-Open)}\\
\Ea;\Tres\vdash P: \ok \effect{\Tprov}{\Ttla}, \Ec 
 \\\hline
\Ea; \Tres\vdash \tcopen m.P: \ok \effect{\Tprov}{\Ttla}, \Ec
 \end{array} 
\end{ruleset}
\caption[Type rules for the capability $\tcin m$.]{Type rules for the
  virtually timed ambients.\label{tab:type.rules} }
\end{table}


Now consider the capability rules.  In
\rn{T-Consume}, the resource consumption is a requirement to the
environment, expressed by increasing the coeffect to $\Tres+1$.  Since
the process requires a time slice, it is counted among the
currently hosted processes. If it was already counted as a timed
process, $\Ttla$ remains unchanged, but since it could have been
untimed, we let $\Ttla'=\max\{\Ttla,1\} $. 



Rule~\rn{T-In} derives an assumption about ambient $m$ under which the
movement $\tcin m.P$ can be typed. Since the movement involves all
processes co-located with $\tcin m.P$, the rule depends on the
resource contract of $\Tcur$, the ambient in which the current process
is located. The rule has a premise expressing that if $P$ can be typed
with a resource contract $T$ for $m$, then $\tcin m.P$ can be typed
with the resource contract $T'$ for $m$. In addition, the hosting
capacity $\Tcap'$ of $\Tcur$ and $\Tcur$ itself are added to the
assumed currently hosted processes $\Tts $ of the premise. The premise
$\Tcap\times \Tres\leq \Tresj$ expresses that the required resources
$\Tres$ must be within the resource capacity $\Tresj$ if scheduled to
all processes within the hosting capacity $\Tcap$ of $m$.  The effect
and co-effect carry over directly from the premise, as the movement
does not modify the required or provided resources or subambients of
$P$.
In contrast, rules~\rn{T-Open} and \rn{T-Out} simply preserve the
co-effect and effect of its premise, since the actual movement is
captured by the worst-case assumption in \rn{T-Amb}.


\begin{example}[Typing of in-capabilities]\label{ex2}
  We revisit Example~\ref{ex1} to illustrate the typing of
  $\tamb{\textit{cloud\,}}{\0} \mid \tamb{vm}{\tcin
    \textit{cloud\,}.\tcreseat. \0}$.  From \rn{T-Zero}
  and \rn{T-Consume}, we get
   $\emptyset;1\vdash \tcreseat. \0:\ok \effect{0}{1};\emptyset$.
   The $\tcin$-capability will move the ambient containing this process,
   which is captured by $\Tcur$ in the typing environment. Let us type $\Tcur$ by $T=\langle
   1,1,1\rangle$. In this case \textit{cloud} will need a hosting capacity if
   at least $2$, so let us type \textit{cloud} by $T'=\langle 2, 2, 2 \rangle$.
 Then, from 
   \rn{T-In}, we get
    $$\textit{cloud}:T',\Tcur:T;1\vdash \tcin \textit{cloud\,}.\tcreseat. \0:\ok
    \effect{0}{1};\emptyset.$$
    By \rn{T-Amb}, we get
    $\textit{cloud\,}:T';1\vdash  \tamb{vm}{\tcin~\textit{cloud\,}.\tcreseat. \0}:\ok \effect{0}{2};vm: T$.
    Similarly,
    $\emptyset;2\vdash \tamb{\textit{cloud\,}}{\0}: \ok \effect{0}{1};
    \textit{cloud}:\langle 2, 2, 0 \rangle$
and \rn{T-Par} gives us
    $$\emptyset;2\vdash \tamb{\textit{cloud\,}}{\0}
\parallel  \tamb{vm}{\tcin
      \textit{cloud\,}.\tcreseat. \0}
    : \ok \effect{0}{3}; vm: T, \textit{cloud}:T';$$
\end{example}

\begin{example}[Typing of open-capabilities]
  We consider the typing of a process
  $\tamb{\textit{cloud\,}}{\tacopen\; vm.\0 \parallel
    \tamb{vm}{\tcreseat.\0}}$. From \rn{T-Zero}
  and \rn{T-Consume}, we get
   $\emptyset;1\vdash \tcreseat. \0:\ok \effect{0}{1};\emptyset$. Let
   $vm$ have type $T=\langle 1,1,1\rangle$. Then, by \rn{T-Amb}, 
$$\emptyset;1\vdash \tamb{vm}{\tcreseat. \0}:\ok \effect{0}{2}; vm:T.$$
By \rn{T-Zero}, \rn{T-Open} and \rn{T Sub}, we have
 $\emptyset;1 \vdash \tacopen\; vm.\0 :\ok
 \effect{0}{0};\emptyset$. By \rn{T-Par}, we obtain
 $\emptyset;1\vdash \tacopen\; vm.\0 \parallel
    \tamb{vm}{\tcreseat.\0}: \ok \effect{0}{2}; vm:T$.
 Let \textit{cloud} have type $T'=\langle 2, 2, 2 \rangle$.
By \rn{T-Amb}, we get
 $$\emptyset;2\vdash \tamb{\textit{cloud\,}}{\tacopen\; vm.\0 \parallel
    \tamb{vm}{\tcreseat.\0}}: \ok \effect{0}{3}; vm:T, \textit{cloud}:
  T'.$$
\end{example}

\begin{example}[Typing of out-capabilities]
  We consider the typing of a process
$$  \tamb{\textit{cloud}}{\tamb{vm}{\tacout
    \textit{cloud}. \tcreseat.\0} \parallel \0}$$
By \rn{T-Zero} and \rn{T-Consume}, we have
$\emptyset; 1 \vdash\tcreseat.\0:\ok\effect{0}{1}; \emptyset$, and by
\rn{T-Out} we get
$$
\emptyset; 1 \vdash\tacout \textit{cloud}. \tcreseat.\0: \ok\effect{0}{1}; \emptyset
$$
Let $T=\langle 1,1,1\rangle$. We can type $vm$ by
$$
\emptyset; 1 \vdash\tamb{vm}{\tacout \textit{cloud}. \tcreseat.\0}:\ok\effect{0}{2}; vm:T
$$
and, with $T'= \langle 2,2,2\rangle$, we get
$$ \emptyset; 2 \vdash \tamb{\textit{cloud}}{\tamb{vm}{\tacout
    \textit{cloud}. \tcreseat.\0} \parallel \0}: \ok\effect{0}{3};
vm:T, \textit{cloud}:T'$$
\end{example}

\begin{example}[Failure of type checking] \label{vta:types:ex:fail}
  Type checking fails if the provisioning of resources for an incoming
  ambient in a timely way cannot be statically guaranteed. This can
  occur for different reasons.  One reason is that an ambient may
  lack \emph{sufficient hosting capacity}  to take in the processes that want to
  enter. Let $T'=\langle 2, 2, 2 \rangle$ as before and consider again the process
  $\tamb{\textit{cloud\,}}{ \0} \mid \tamb{vm}{ \tcin
    \textit{cloud\,}.\tcreseat. \0}$ from Example~\ref{ex2}.  Now
  assume a second virtual machine 
  $\tamb{vm_2}{ \tcin \textit{cloud\,}.\tcreseat. \0}$ which aims to
  enter the \textit{cloud} ambient, resulting in the parallel process
 $$\tamb{\textit{cloud\,}}{ \0} \mid \tamb{vm}{ \tcin
   \textit{cloud\,}.\tcreseat. \0} \parallel \tamb{vm_2}{ \tcin
   \textit{cloud\,}.\tcreseat. \0}$$ We can type $vm_2$ similarly to
 $vm$ in Example~\ref{ex2}.:
  $$\textit{cloud\,}:T';1\vdash  \tamb{vm_2}{\tcin
    \textit{cloud\,}.\tcreseat. \0}:\ok \effect{0}{2};vm_2: T.$$
 In contrast to Example~\ref{ex2}, the hosting capacity for
  \textit{cloud} in $T'$ cannot accommodate both $vm$ and $vm_2$;
  type checking fails when giving \textit{cloud} resource
  contract $T'$.
  
Another reason is that 
he resource contract of \textit{cloud} may have a \emph{too low
 resource capacity}. Consider a third virtual machine
  $\tamb{vm_3}{ \tcin\;
    \textit{cloud\,}.\tcreseat.\tcreseat.\tcreseat. \0}$ which can be
  typed with the resource contract  $ \langle 3,1,1\rangle$ for $vm_3$.
  Again, type checking fails if \textit{cloud} were given the resource
  contract $T'$, since the resource capacity of \textit{cloud} must
  here be
  at least $6$ with hosting capacity $2$.
\end{example}

\begin{example}[Capacity of an ambient] 
Assume that the process 
$$
\tamb{n_1}{ \tcin\; m.P_1} \parallel
  \tamb{n_2}{\tcin\; m.P_2} \parallel \tamb{m}{Q}
$$
is well-typed.
 Let
  $T_1= \Tamb{\Tresj}{\Tcap}{\Tts_1}{}$, 
  $T_2= \Tamb{\Tresj}{\Tcap}{\Tts_2}{}$ and
  $T_3= \Tamb{\Tresj}{\Tcap}{\Tts_3}{}$
be resource contracts such
  that
 $$m:T_i; \Tres_i \vdash  \tamb{n_i}{\tcin m.P_i}: \ok
 \effect{\Tprov_i}{\Ttla_i}; \Delta_i$$
 for $i\in\{1,2\}$, and $\emptyset;\Tres_3\vdash \tamb{m}{Q}:\ok
 \effect{\Tprov_3}{\Ttla_3}; m:T_3$. 
Let $r_{12}=\max(r_1,r_2)$ and 
 $T_{12}=\Tamb{\Tresj}{\Tcap}{\Tts_1\oplus \Tts_2}{}$. Since
 $\tamb{n_1}{ \tcin\; m.P_1} \parallel \tamb{n_2}{\tcin\; m.P_2} $ is
 well-typed, we have $\Tts_1\oplus \Tts_2\leq \Tcap$ and, by \rn{T-Par},
 $$m:T_{12}; \Tres_{12}
\vdash \tamb{n_1}{ \tcin\; m.P_1} \parallel
  \tamb{n_2}{\tcin\; m.P_2}:\ok \effect{\Tprov_{12}}{\Ttla_{12}}; \Ec_{12}
$$
where $\Tprov_{12}=\Tprov_1+\Tprov_2$, $\Ttla_{12}=\Ttla_1+\Ttla_2$
and $\Delta= \Delta_1,\Delta_2$.
By applying \rn{T-Par} again, we get
$$\emptyset; \Tres \vdash
\tamb{n_1}{ \tcin\; m.P_1} \parallel
  \tamb{n_2}{\tcin\; m.P_2} \parallel \tamb{m}{Q}:\ok \effect{\Tprov}{\Ttla}; m:T,\Delta
$$
where 
 $\Tres=\max\{\Tres_{12}, \Tres_{3}\}$, 
$\Tprov= \Tprov_{12}+\Tprov_{3}$,
$\Ttla= \Ttla_{12}+\Ttla_{3}$ and
$T= \Tamb{\Tresj}{\Tcap}{\Tts_{12}+\Tts_3}{}$.
Thus, the weakest resource contract which types $m$ and allows both
$n_1$ and $n_2$ to enter,  will have $\Tcap= \Tts_{12}+\Tts_3$ and
$\Tresj=\Tcap\times \Tres$.
\end{example}


\section{Soundness of Resource Management}
\label{sec:coeff.subred}
The soundness of resource management can be perceived similarly
to that of message exchange \cite{CGD01}. We prove a subject
reduction theorem, stating that the number of resources required to
terminate a process is preserved under reduction.

\begin{theorem}[Subject Reduction]
  \label{subred}
  \index{subject reduction} Assume
  $\Gamma, \Tres\p P:\ok \effect{\Tprov}{\Ttla}; \Delta$ and
  $ P \red Q$, then there are environments $\Gamma' \Tleq \Gamma$ and
  $\Delta'\Tleq \Delta$ such that
  $\Gamma', \Tres '\vdash Q:\ok \effect{\Tprov '}{\Ttla'};\Delta'$ and
  $\Tres '\leq \Tres$ or $\Tres '=\Tres \land \Tprov '\geq \Tprov$.
\end{theorem}



\begin{proof}
  By induction on the derivation of $P \red Q$ (For details, see Appendix~\ref{app:proof}).
\end{proof} 


Further, we prove a progress theorem, which shows that a well-typed process
which receives the approximated number of resources from its environment
will not get stuck because of missing resources. We use the contextual
variant of barbing from Def.\ \ref{definition:coeff.barbs} to characterize
a situation where inside the process, there is a sub-process in need of a
tick to proceed, be it an unserved ambient or a process guarded by a
$\tickins$-capability. \iflong Note that in the formulation of the progress
property, the process $P$ about which the type system asserts resources
needs, in particular $\Tres$, needs to be embedded into a outermost,
surrounding ambient; only then, the corresponding scheduler can assure that
the $\Tres$ available ticks are handed out. In particular, progress could
not be guaranteed in the case where $P = \tamb{\frozen{m}}{P'}$, i.e., in a
situation, where the process $P$ is of the form of an already served
ambient. The notion of being served (or not) makes only sense in the context
of a surrounding ambient and its associated scheduling mechanism.\fi

\begin{theorem}[Tick progress]
  Assume 
$\Ea;\Tres \vdash P:\ok \effect{\Tprov}{\Ttla};\Ec$
and let
$Q = \tamb{\frozenornot{n}}{P \parallel \tickouts \parallel
  \ldots \parallel \tickouts}$,
where $P$ is running in parallel with $\Tres$ occurrences of
$\tickouts$ inside some enclosing ambient. If
$Q \sbarbsonwith{\tctickin}{\C}$ for some context $\C$, then 
$Q \red{} Q'$ for some process $Q'$.
\end{theorem}

\begin{proof}
  This follows from the definition of the typing rules. If $P$
  contains the subprocess $\tcreseat.P'$ it follows from the typing
  rule for the consume capability that $\Tres \geq 1$. From the other
  typing rules it follows that $\Tres$ cannot be reduced to
  $\Tres < 1$ and the number of resources is sufficient to trigger the
  reduction $\tcreseat.P' \red P'$. Thus, $P$ can reduce to $Q $ and
  $\Tres>0$.
\end{proof}





With the properties of subject reduction and progress the type system
guarantees the soundness of resource management.
\begin{corollary}[Soundness]
  The type system guarantees the soundness of resource management,
  i.e., the transitive closure of the progress result holds.
\end{corollary}


\section{Related Work} 
\label{sec:coeff.related}
Gordon proposed a simple formalism for virtualization loosely based on
mobile ambients \cite{GordonVirtual}. The calculus of virtually timed
ambients \cite{johnsen.steffen.stumpf:calculusvta-wadt} stays closer
to the syntax of the original ambient calculus, while including
notions of time and resources. Our model of resources as processing
capacity over time builds on deployment components
\cite{johnsen15jlamp,ABHJ14}, a modelling abstraction for cloud
computing in ABS \cite{johnsen10fmco}. Compared to virtually timed
ambients, ABS does not support nested deployment components nor the
timed capabilities of ambients.  Timers have been studied for mobile
ambients in \cite{AC07b}. In this line of work, timers, which are
introduced to express the possibility of a timeout, are controlled by
a global clock. In contrast, the schedulers in our work recursively
trigger local schedulers in subambients which define the execution
power of the nested virtually timed ambients. Modelling timeouts is a
straightforward extension of our work.  The calculus of virtually
timed ambients presented here differs from earlier papers
\cite{johnsen.steffen.stumpf:calculusvta-wadt,johnsen.steffen.stumpf:calculusvta-jlamp}
by assuming uniform time and by the use of freezing and unfreezing
operations, which allow a significantly simpler formulation of the
calculus. The behavior of the original calculus, with non-uniform
time, can be recaptured by modifying the rule
$\tickins \red \tickouts$ to cater for different numbers of input and
output ticks, and to contextualize the rule for specific ambients. For
the virtually timed ambients with non-uniform time, a modal logic with
somewhere and sometime modalities has been developed
\cite{johnsen18ictac} to express aspects of reachability for
these ambients. Whereas this work can express more complex
properties of a given process, the logic cannot express properties for
all processes, in contrast to the contract-based type system presented
in this paper.


A type system for the (originally untyped) ambient calculus was
defined in \cite{CGD01}; this type system is mainly concerned with the
use of groups to control communication and mobility. For
communication, a basic type of an ambient captures the kind of
messages that can be exchanged within. For mobility, the type system
controls which ambients can enter. In a more traditional setting of
sequential languages, types are often enriched with effects to capture
the aspects of of computation which are not purely functional. In
process algebra, session types have been used to capture communication
in the $\pi$-calculus. Orchard and Yoshida have shown that effects and
session types are similar concepts as they can be expressed in terms
of each other \cite{OY16}. Session types have been defined for boxed
ambients in \cite{GC000} and behavioral effects for the ambient
calculus in \cite{A08}, where the original communication types by
Cardelli and Gordon are enhanced by movement behavior. This is
captured with traces, the flow-sensitivity hereby results from the
copying of the capabilities in the type.  Type-based resource control
for resources in the form of locks has been proposed for process
algebras in general~\cite{igarashi2002resource} and for the
$\pi$-calculus in
particular~\cite{kobayashi2006resource,kobayashi2010hybrid}.

The idea
of assumptions and commitments (or relies and guarantees) is quite
old, and has been explored in various settings, mainly for
specification and compositional reasoning about concurrent or parallel
processes (e.g., \cite{abadi.lamport:conjoining,jones:tentative,lamport:modules,misra.chandy:networks,stark:prooftechnique}).
Assumption commitment style type systems have previously been used for
multi-threaded concurrency \cite{abraham*:futures,abraham.gruener.steffen:heapabstraction-cie}; the resources
controlled by the effect-type system there are locks and a general
form of futures, in contrast to our work.

To capture
how a computation depends on an environment instead of how the
computation affects it, Petricek, Orchard and Mycroft suggest the term
\emph{coeffect} as a notion of \emph{context-dependent} computation
\cite{POM13,POM14}. Dual to effects, which can be modeled monadically,
the semantics of coeffects is provided by indexed
comonads~\cite{katsumata:parametric,uustalu.vene:comonadic}. 
We use coeffects to control time and resources.  An approach to
control timing via types can be found in
\cite{berger.yoshida:timeddistributed}, which develops types and typed
timers for the timed $\pi$-calculus.  Another approach to resource
control without coeffects can be found in \cite{HENNESSY200282}, which
proposes a type system to restrict resource access for the distributed
$\pi$-calculus. In \cite{TZH02} a type system for resource control for
a fragment of the mobile ambients is defined by adding capacity and weight
to communication types for controlled ambients. Simplified
non-modifiable mobile ambients with resources, and types to control
migration and resource distribution are proposed in
\cite{godskesen2002calculus}. Another fragment of the ambient calculus, finite
control ambients with only finite parallel composition, are covered in
\cite{charatonik2002finite}. Here the types are a bound to the number of
allowed active outputs in an ambient.


\section{Concluding Remarks}
\label{sec:coeff.conclusion}

Virtualization opens for new and interesting models of computation by
explicitly emphasizing deployment and resource management.  This paper
introduces a type system based on resource contracts for virtually
timed ambients, a calculus of hierarchical locations of execution with
explicit resource provisioning.  Resource provisioning in this
calculus is based on virtual time, a local notion of time reminiscent
of time slices provisioned by operating systems in the context of
nested virtualization.  The proposed assumption-commitment type system
with effects and coeffects enables static checking of timing and
resource constraints for ambients and gives an upper bound on the
resources used by a process.
The type system supports subsumption, which allows relating
different types, e.g. weaker types, to each other. 
We show that the proposed type system is sound in terms of
subject reduction and a progress properties.  Although these are core
properties for type systems, the results are here given for a
non-standard assumption-commitment setting in an operational
framework.
The type system further provides reusable properties as it supports
abstraction and the results would also hold for other operational
accounts of fair scheduling strategies.  The challenge of how to
further generalize the distribution strategy and type system for,
e.g., earliest deadline first or priority-based scheduling policies,
remains.

The virtually timed ambients used for the models in this paper extend
the basic ambient calculus without channel communication.  Introducing
channels would lead to additional synchronization, which could
potentially be exploited to derive more precise estimations about
resource consumption. Such an extension would be non-trivial as
the analysis of the communication structure would interfere with
scheduling.


\newpage
\appendix
\section{Proof of Theorem 1}\label{app:proof}


The proof proceeds by cases over the reduction rules of tables~\ref{tab:redsteps.simple.1}
and \ref{tab:redsteps.simple.2}. In each case, we assume that the
pre-state is well-typed and show that this assumption allows us to
construct a type derivation for the post-state.

\paragraph{Case \rn{R-In}.}
Assume that $\tamb{\meta{n}}{\tacin~m.P_1\parallel P_2} \parallel
\tamb{\meta{m}}{Q}$ is well-typed. Consequently, for some values
$\Tresj$ and $\Tcap$ we have types
$T_1=\langle \Tresj,\Tcap,\Tts_1 \rangle$,
$T_2=\langle\Tresj,\Tcap,\Tts_2 \rangle$  and 
$T_3=\langle \Tresj,\Tcap,\Tts_3 \rangle$ such that
the following assumptions hold:\\
(1) $\Gamma_1, m:T_1; r_1\vdash P_1:\ok \effect{p_1}{s_1}; \Delta_1$, 
(2) $\Gamma_2, m:T_2; r_2\vdash P_2:\ok \effect{p_2}{s_2}; \Delta_2$
and
(3) $\Gamma_3,\Tcur:T_3; r_3\vdash Q:\ok \effect{p_3}{s_3};\Delta_3$.

\medskip

\noindent
It follows from Assumption~3 by \rn{T-Amb} that
$$\Gamma_3; \Tresj\vdash  \tamb{m}{Q}:\ok \effect{0}{s_3+1};m:T_3,\Delta_3.$$
Let $n$ be typed by $\langle\Tresj',\Tcap\,',\Tts'\rangle$ and
 let $T_1'= \langle \Tresj,\Tcap,\Tts_1+\Tcap\,'+1 \rangle$.  It
follows from Assumption~1 by \rn{T-In} that
$$\Gamma_1,m:T_1', r_1\vdash \tacin{m}.P_1:\ok\effect{p_1}{s_1};\Delta_1$$
Let $T'_{12}=T_1'\oplus T_2= \langle \Tresj,\Tcap,\Tts_1+\Tcap\,' +1+\Tts_2\rangle$.
Since $\Gamma_1 : \ok$, we know that
$\Tcap\geq \Tts_1+\Tcap\,' +1+\Tts_2$. Let
$r_{12}=\max\{r_1,r_2\}$. From assumptions 1 and 2, we get 
from \rn{T-Par} that
\begin{equation}\label{eq1}
\Gamma_1\oplus \Gamma_2, m:T'_{12};r_{12},\vdash
\tacin{m}.P_1\parallel P_2 :\ok\effect{p_1+p_2}{s_1+s_2};\Delta_1,\Delta_2
\end{equation}
Let $s_{12}=s_1+s_2+1$. Since $r_{12}\times\Tcap\,'\leq \Tresj'$, we get from \rn{T-Amb} that
$$\Gamma_1\oplus\Gamma_2, m:T'_{12};\Tresj',\vdash
\tamb{\meta{n}}{\tacin\ m.P_1\parallel P_2}
:\ok\effect{0}{s_{12}}; n:\langle\Tresj',\Tcap\,',\Tts'\rangle, \Delta_1,\Delta_2.$$
Let $\Gamma= \Gamma_1\oplus \Gamma_2\oplus\Gamma_3$,
$\Delta=\Delta$, 
$T'_{123}=T_{12}'\oplus T_3=\langle \Tresj,\Tcap,
\Tts_1+\Tcap\,' +1+\Tts_2+\Tts_3\rangle$ and $r=\max\{\Tresj,\Tresj'\}$. Then, by \rn{T-Par}, we have 
$$\Gamma;r,\vdash
\tamb{\meta{n}}{\tacin{m}.P_1\parallel P_2} \parallel \tamb{m}{Q}:\ok \effect{p_3}{s_1+s_2+s_3+2}; m:T'_{123},\Delta.$$

\smallskip

\noindent
Now, we show that 
$\tamb{\meta{m}}{Q\ \parallel \tamb{\frozen{n}}{P_1\parallel P_2}}$ 
is well-typed.
From Equation~\ref{eq1}, we get 
$\Gamma_1\oplus \Gamma_2, m:T'_{12};r_{12}\vdash P_1\parallel P_2:\ok\effect{p_1+p_2}{s_1+s_2};\Delta_1,\Delta_2$,
and by \rn{T-Amb}
$$\Gamma_1\oplus \Gamma_2, m:T'_{12};\Tresj'\vdash
\tamb{\frozen{n}}{P_1\parallel P_2}: \ok\effect{0}{s_{12}};n:\langle\Tresj',\Tcap\,',\Tts'\rangle;\Delta_1,\Delta_2$$
It follows by \rn{T-Par} that 
$$\Gamma, m:T'_{123};r\vdash
\tamb{\frozen{n}}{P_1\parallel P_2} \parallel  Q:
\ok\effect{p_3}{s_1+s_2+s_3+2};\Delta$$
Let $r=\max\{ \}$. By \rn{T-Amb}, it follows that
$$\Gamma, m:T'_{123};\Tresj\vdash
\tamb{m}{Q\ \parallel \tamb{\frozen{n}}{P_1\parallel P_2}}: \ok\effect{0}{s_1+s_2+s_3+2};\Delta.$$
Since $\Tresj\leq r$, the case holds.

\bigskip

\paragraph{Case \rn{R-Out}.}
Assume that $ \tamb{\meta{m}}{\tamb{\meta{n}}{\tacout~m.P_1\parallel
    P_2} \parallel Q}$ is well-typed.
Consequently, 
the following assumptions hold:
(1) $\Gamma_1; r_1\vdash P_1:\ok \effect{p_1}{s_1}; \Delta_1$, 
(2) $\Gamma_2; r_2\vdash P_2:\ok \effect{p_2}{s_2}; \Delta_2$
and
(3) $\Gamma_3; r_3\vdash Q:\ok \effect{p_3}{s_3};\Delta_3$.

\medskip

\noindent
From Assumption~1, by \rn{T-Out}, we obtain
$\Gamma_1,; r_1\vdash \tacout{m}.P_1:\ok \effect{p_1}{s_1};
\Delta_1$.
Let 
$r_{12}=\max\{r_1,r_2\}$. It then follows from
Assumption~2, by \rn{T-Par}, that
\begin{equation}\label{eq3}
\Gamma_1\oplus \Gamma_2; r_{12}\vdash \tacout{m}.P_1\parallel P_2:\ok \effect{p_1+p_2}{s_1+s_2};
\Delta_1,\Delta_2.
\end{equation}
Let $T_n=\langle \Tresj_n,\Tcap_n,\Tts_n \rangle$. Then, by \rn{T-Amb}, we obtain
$$\Gamma_1\oplus \Gamma_2; \Tresj_n\vdash \tamb{\meta{n}}{\tacout~m.P_1\parallel
    P_2}:\ok \effect{0}{s_1+s_2+1};n:T_n,
\Delta_1,\Delta_2.$$
Now, let $\Gamma= \Gamma_1\oplus \Gamma_2\oplus \Gamma_3$, $\Delta= \Delta_1,\Delta_2,\Delta:3$,
$r_{123}=\max\{\Tresj_n,r_3\}$ and $s_{123}=s_1+s_2+s_3+1$.
We get from Assumption~3 using \rn{T-Par} that 
$$\Gamma; r_{123}\vdash \tamb{\meta{n}}{\tacout~m.P_1\parallel
    P_2} \parallel Q:\ok \effect{p_3}{s_{123}};n:T_n,
\Delta.$$
Let $T_m=\langle \Tresj_m,\Tcap_m,\Tts_m \rangle$ and $r=\{r_{123},\Tresj_m\}$. From \rn{T-Amb}, 
$$\Gamma; r\vdash \tamb{\meta{m}}{\tamb{\meta{n}}{\tacout~m.P_1\parallel
    P_2} \parallel Q}:\ok \effect{0}{s_{123}+1};m:T_m,n:T_n,\Delta. $$

\smallskip

\noindent
Now, we show that 
$\tamb{\frozen{n}}{P_1\parallel P_2} \parallel \tamb{\meta{m}}{Q}$ 
is well-typed. From Equation~\ref{eq3} and \rn{T-Out}, we know that 
$$\Gamma_1\oplus \Gamma_2; r_{12}\vdash P_1\parallel P_2:\ok \effect{p_1+p_2}{s_1+s_2};
\Delta_1,\Delta_2$$
and, by \rn{T-Amb}, we obtain
$$\Gamma_1\oplus \Gamma_2; \Tresj_n\vdash \tamb{\frozen{n}}{P_1\parallel P_2}:\ok  \effect{0}{s_1+s_2+1};n:T_n,
\Delta_1,\Delta_2.$$
Now, let $T_m'=\langle r_3,s_3,s_3 \rangle$. It follows from
Assumption~3 by \rn{T-Amb} that
$$\Gamma_3; r_3\vdash \tamb{\meta{m}}{Q}:\ok  \effect{0}{s_3+1};m:T_m',
\Delta_3$$
and, by \rn{T-Par}, that 
$$\Gamma; r_{123}\vdash \tamb{\frozen{n}}{P_1\parallel P_2} \parallel \tamb{\meta{m}}{Q}:\ok
\effect{0}{s_{123}+1}; m:T_m',n:T_n,
\Delta.
$$
We know that $r_{123}\leq r$ and $m:T_m',n:T_n, \Delta \Eleq
m:T_m,n:T_n,\Delta$, which closes the case.

\paragraph{Case \rn{R-Open}.}
Assume that $\tacopen{n}.P_1 \parallel \tamb{\meta{n}}{P_2}$ is well-typed.
Consequently, 
the following assumptions hold:
(1) $\Gamma_1; r_1\vdash P_1:\ok \effect{p_1}{s_1}; \Delta_1$, and
(2) $\Gamma_2; r_2\vdash P_2:\ok \effect{p_2}{s_2}; \Delta_2$.

\medskip

\noindent
It follows from Assumption~1 that $\Gamma_1; r_1\vdash \tacopen{n}.P_1
:\ok \effect{p_1}{s_1}; \Delta_1$, Let $n$ be typed by some resource
contract $T=\langle \Tresj,\Tcap,\Tts \rangle$. Then, from
Assumption~2, \rn{T-Amb} gives us
$$\Gamma_2; \Tresj\vdash  \tamb{\meta{n}}{P_2}:\ok \effect{0}{s_2+1};
n:T, \Delta_2$$
where we know that $r_2\times \Tcap \leq \Tresj + p_2$ and $s_2\leq
\Tcap$. Let $r=\max\{ r_1,\Tresj\}$.
It now follows from \rn{T-Par} that 
$$\Gamma_1\oplus \Gamma_2; r\vdash \tacopen{n}.P_1 \parallel \tamb{\meta{n}}{P_2}:\ok \effect{p_1}{s_1+s_2+2};
n:T, \Delta_2
$$

\smallskip

\noindent
Now, we show that 
$P_1 \parallel \frozen{P_2}$
is well-typed. 
Let  $r_{12}=\max\{ r_1,r_2\}$.
From assumptions~1 and 2, we know by \rn{T-Par} that
$$\Gamma_1\oplus \Gamma_2; r_{12} \vdash  P_1 \parallel \frozen{P_2}:\ok \effect{p_1+p_2}{s_1+s_2};
n:T, \Delta_2.$$
We know that $r_{12}\leq r$ and $0\leq p_1+p_2$, which closes the case.

\paragraph{Case $\tickins \red{} \tickouts$.}
This case is immediate as we have (\rn{T-Tick1})
$\emptyset;0\vdash \tickins :\ok\effect{1}{0}; \emptyset$
and (\rn{T-Tick2})
$\emptyset;0\vdash \tickouts :\ok\effect{1}{0}; \emptyset$.

\paragraph{Case $ \dunfrozen{\tcreseat}.P \red{} \tickins.P$.}
Assume that $\Gamma;r\vdash P:\ok\effect{p}{s}; \Delta$.  Let
$s'=\max\{ s,1\}$.  This case is immediate as we have, by
\rn{T-Consume1}, $\Gamma;r+1\vdash
\dunfrozen{\tcreseat}.P:\ok\effect{p}{s'}; \Delta$ and, by
\rn{T-Consume2} $\Gamma;r+1\vdash \tickins.:\ok\effect{p}{s'};
\Delta$.

\paragraph{Case $ \tickouts \parallel \tickins.P \red{} \freeze{P}$.}
Assume
$\Gamma;r\vdash P:\ok\effect{p}{s}; \Delta$.
Then by \rn{T-Consume}, 
$\Gamma;r+1\vdash  \tickins.P:\ok\effect{p}{s};\Delta$
and, by \rn{T-Par}, 
$$\Gamma;r+1\vdash   \tickouts \parallel \tickins.P:\ok\effect{p+1}{s};\Delta$$
From the assumption,  we have $\Gamma;r\vdash \freeze{P}:\ok\effect{p}{s};\Delta$
and since $r\leq r+1$ the case holds.

\paragraph{Case $ \tickouts \parallel \tamb{\dunfrozen{n}}{P}\red{}
  \tamb{\frozen{n}}{\tickins \parallel P}$.}
Let $n$ have resource contract $\langle \Tresj,\Tcap,\Tts \rangle$ and assume
$$\Gamma,r\vdash P:\ok\effect{p}{s};\Delta.$$ 
Since $\tamb{\dunfrozen{n}}{P}$ is well-typed, we have 
$\Tres \times \Tcap\leq \Tresj+p$
such that 
$$\Gamma;\Tresj\vdash \tamb{\dunfrozen{n}}{P}:\ok\effect{p}{s+1}; n:\langle \Tresj,\Tcap,\Tts \rangle;\Delta.$$
We then have that 
$$\Gamma; r\vdash \tickins \parallel P:\ok\effect{p+1}{s+1};\Delta$$ 
Consequently
$$\Gamma,\Tresj-1\vdash  \tamb{\frozen{n}}{\tickins \parallel
  P}:\ok\effect{p+1}{s+1}; n:\langle \Tresj-1,\Tcap,\Tts
\rangle, \Delta$$
and, since $\Tresj-1\leq \Tresj$, the case holds.

\paragraph{Case $ \tamb{n}{P} \red{} \tamb{n}{\unfreeze{P}}$ (New
  round).}
Let $\Gamma \vdash n: \langle \Tresj,\Tcap,\Tts \rangle$.  We can
assume $\Gamma, r\vdash P : \ok\effect{p}{s}$ such that
$\Gamma, \Tresj\vdash \tamb{n}{P} :: \ok\effect{0}{s+1}$. It
follows that
$\Gamma, r\vdash \unfreeze{P} : \ok\effect{p}{s}$ and
consequently
$\Gamma, \Tresj\vdash  \tamb{n}{\unfreeze{P}} :: \ok\effect{0}{s+1}$.

\qed



\end{document}
